\documentclass[conference]{IEEEtran}

\usepackage{amsmath,amsthm,amssymb,xspace,stmaryrd,hyperref,tikz,bm,setspace,threeparttable,dblfloatfix,fancyhdr}
\usetikzlibrary{patterns,positioning,snakes,calc}
\usepackage{todonotes}
\usepackage{cleveref}
\newtheorem{thm}{Theorem}
\newtheorem{lem}{Lemma}
\newtheorem{cor}{Corollary}
\def\F{\ensuremath{\mathbb{F}}\xspace}

\def\K{\ensuremath{\mathbb{K}}\xspace}
\def\GRS{\ensuremath{\mathrm{RS}}\xspace}
\def\code{\ensuremath{\mathcal{C}}\xspace}
\newcommand{\seq}[1]{\ensuremath{\llbracket #1 \rrbracket}\xspace}
\newcommand{\floor}[1]{\ensuremath{\left\lfloor #1 \right\rfloor}\xspace}
\def\mix{\ensuremath{\mathrm{Mix}}\xspace}
\def\RS{\GRS}
\def\e{\mathrm{e}}
\newcommand{\best}[1]{\textcolor{green!45!black}{#1}}

\def\BibTeX{{\rm B\kern-.05em{\sc i\kern-.025em b}\kern-.08em
    T\kern-.1667em\lower.7ex\hbox{E}\kern-.125emX}}

\title{Unraveling codes: fast, robust, beyond-bound error correction for DRAM}

\author{Mike Hamburg, Eric Linstadt, Danny Moore, Thomas Vogelsang
\\ Rambus Inc.
\\ \texttt{\{mhamburg,dmoore,elinstadt,tvogelsang\}@rambus.com}}

\newcommand\chipkill{CC\xspace}
\newcommand\doublechipkill{DCC\xspace}
\newcommand\bbck{BBCC\xspace}

\begin{document}
\maketitle
\thispagestyle{plain}
\pagestyle{plain}

\begin{abstract}
    Generalized Reed-Solomon (RS) codes are a common choice for efficient, reliable error correction in memory and communications systems.  These codes add $2t$ extra parity symbols to a block of memory, and can efficiently and reliably correct up to $t$ symbol errors in that block.  Decoding is possible beyond this bound, but it is imperfectly reliable and often computationally expensive.
    
    Beyond-bound decoding is an important problem to solve for error-correcting Dynamic Random Access Memory (DRAM).  These memories are often designed so that each access touches two extra memory devices, so that a failure in any one device can be corrected.  But system architectures increasingly require DRAM to store metadata in addition to user data.  When the metadata replaces parity data, a single-device failure is then beyond-bound.

    An error-correction system can either protect each access with a single RS code, or divide it into several segments protected with a shorter code, usually in an Interleaved Reed-Solomon (IRS) configuration.  The full-block RS approach is more reliable, both at correcting errors and at preventing silent data corruption (SDC).  The IRS option is faster, and is especially efficient at beyond-bound correction of single- or double-device failures.

    Here we describe a new family of ``unraveling'' Reed-Solomon codes that bridges the gap between these options.  Our codes are full-block generalized RS codes, but they can also be decoded using an IRS decoder.  As a result, they combine the speed and beyond-bound correction capabilities of interleaved codes with the robustness of full-block codes, including the ability of the latter to reliably correct failures across multiple devices.  We show that unraveling codes are an especially good fit for high-reliability DRAM error correction.
\end{abstract}

\section{Introduction}
Errors in computer memory have been a concern for decades, and this problem is more important now than ever.  Decreasing feature size leads to increased error rates per gigabyte per unit time \cite{meza2015revisiting},
and these rates must then be multiplied by the huge amount of memory in warehouse-size data centers \cite{barroso2019datacenter,beigi2023systematic,schroeder2009dram}.  Server-grade memory controllers mitigate this problem by storing extra information in memory for use in an error-correcting code (ECC).

Dynamic Random Access Memory (DRAM) is organized into modules, each containing multiple DRAM devices, several of which are accessed for each read or write operation.  Each access to the module is spread over parallel data lines (called DQ; typical devices have 4 DQ) and over time (burst length).  Most memory errors occur when a component of a DRAM device fails, either transiently or persistently, possibly affecting multiple bits of memory \cite{beigi2023systematic,cheng2022depth,li2022correctable,meza2015revisiting,schroeder2009dram}.  Devices are designed to limit the damage from failures: for example, DDR5 devices include on-chip single-bit error correction, and they are also designed so that a partial-row failure affects only one or two DQs (a \emph{bounded fault} \cite{criss2020improving}).  But some errors affect the entire device: for example, \cite[Table 1]{beigi2023systematic} reports that in DDR4 memory from two vendors, 9.6\% and 16.7\% of errors affected more than a partial row, possibly up to all the device's data for that access.  So it is important that ECC can correct a complete single-device failure.  This feature is known as \emph{single-device data correction} or by IBM's trademark \emph{Chipkill} \cite{dell1997white}; here we will call it \emph{Chipkill-correct} (\chipkill).  A decoder supporting \chipkill might also be capable of correcting errors across multiple devices.  Multi-device errors can be caused, for example, by read-disturb effects such as Rowhammer attacks \cite{kim2014flipping}.

Reed-Solomon codes \cite{reed1960polynomial} correct errors in units of multi-bit symbols.  Such symbol-oriented codes are a natural fit for DRAM's error patterns: symbols can be aligned with a DQ so that an error affecting multiple bits on one DQ can be corrected by a single-symbol correction \cite{gong2018duo, kim2015bamboo}.  A symbol-oriented code with $2t$ parity symbols can correct up to $t$ symbol errors.  According to the Komamiya-Singleton bound \cite{komamiya1953singleton}, it is not possible to correct more symbol errors with perfect reliability: error correction that attempts to do this is said to be \emph{beyond-bound}, and will always have a nonzero failure probability.\footnote{This bound assumes a symmetric channel model, meaning that any symbol value is equally likely to be corrupted into any other symbol value.  This is not strictly the case for DRAM, where depending on the implementation, certain bits are more likely to flip from 0 to 1 and others from 1 to 0.  If the memory controller knows which bits are which, a specialized code can exceed the Komamiya-Singleton bound \cite{9923862}.  However, these asymmetric codes are usually less efficient and less flexible than symmetric codes such as Reed-Solomon.}

To achieve \chipkill, a designer sets $t$ to the number of symbols sent by each device, necessitating two extra devices to store the $2t$ parity symbols.  This principle is used in typical configurations of both DDR4 and DDR5 server memory \cite{JESD79-4,JESD79-5}.  In DDR4 each operation accesses 18 devices, exchanging 4 bytes (4B) per device for a total of 72B.  In DDR5 each operation accesses 10 devices, exchanging 8B per device for a total of 80B.  Typically the user data is a 64B cache line in both cases, leaving 8B or 16B for parity respectively, and in both cases allowing for a full device to be corrected within-bound.  But more recently, e.g., in the Intel Skylake and Cascade Lake server generations \cite{li2022correctable}, some of the extra bits are used as metadata for other purposes than error correction: to store ownership, security information, and cache coherency state; to mark failed areas of DRAM, etc.  This leaves less room for parity, so that beyond-bound Chipkill-correct (\bbck) is required\footnote{Alternatively, extra parity symbols could be written elsewhere \cite{udipi2012lot,jian2013low}, but this costs performance and energy.}.  \bbck cannot be perfectly reliable, but it can correct the vast majority of single-device errors.

The goal of a more powerful decoding algorithm must be balanced against implementation constraints. The implementation needs to achieve low-latency encoding and decoding\footnote{Because memory errors are infrequent, it is possible to use a decoder which detects errors quickly but is slow to correct them.  However, this adds a complexity cost, since such a decoder can drop transactions and/or return them out of order.  Some systems must also maintain good performance even when one device develops a persistent fault.}, small area, and low power consumption.  Beyond-bound decoding is more complex than within-bound decoding, and some approaches rely in part on exhaustive search when errors cannot be decoded analytically \cite{gong2018duo,sudan1997decoding}.  Such a system might correct only single-device errors, or might combine this and within-bound correction of multi-device errors (e.g.\ up to all 4 DQ of a single device, or 3 DQ across multiple devices).  Or, by limiting a search to single-device errors only, the search only needs one computation per device, which may be done in parallel \cite{https://doi.org/10.48550/arxiv.2301.07271}. In higher-reliability systems, beyond-bound double-device correction may be required; in that case, searching for failed devices is quadratically more expensive.

Error-correcting codes should also be designed to reduce the probability of silent data corruption (SDC), which causes insidious application-level problems \cite{dixit2021silent}.  SDC can occur when there are more errors than the code can correct or reliably detect, and in attempting to correct the decoder instead introduces more errors.  
Codes with a larger symbol size and more parity symbols have a lower risk of SDC, because it is less likely for a randomly corrupted data block to be near a valid codeword, but these larger codes also have a higher implementation complexity.

\subsection{Outline and main contributions}

In \cref{concepts} we discuss beyond-bound decoding and interleaved Reed-Solomon codes, which form the basis of our work.  We then introduce our main contributions:
\begin{itemize}
    \item We propose a new type of Reed-Solomon code, which we call \emph{Unraveling Reed-Solomon} (URS) codes, to address the concerns described above.  These URS codes provide robust beyond-bound protection against single- or double-device failure with a very low probability that such a failure is uncorrectable, and can also correct multi-device errors with a low probability of miscorrection.
    \item We rigorously analyze URS codes in \cref{mathcontributions}, and show how to construct, encode and decode them.
    \item We show that they can be implemented with a complexity appropriate to their smaller symbol size (e.g.\ 8 bits), while having the SDC characteristics of a code with large symbol size (e.g.\ 16 or 32 bits).
    \item We implement a high-reliability URS encoder / decoder in a DDR5 memory controller, suitable for modern datacenters.  This is described in \cref{application}.
\end{itemize}
Finally, we compare our approach to other error correction techniques in \cref{comparison}.

\section{Concepts and Terminology} \label{concepts}

Error-correcting codes add redundant bits or symbols to the data to be protected, forming a
\emph{codeword}, and then can use this redundancy to correct a small
number of errors in that codeword.  We focus on generalized Reed-Solomon codes $\GRS(q; n,k)$, which are symbol-oriented linear block codes.  These codes encode $k$ data symbols from a $q$-element finite field $\F := \F_q$ as an $n$-element codeword.  Codewords are only the valid $n$-element sequences: they are a vector subspace of the space $\F^n$ of all $n$-element sequences, which we call \emph{blocks}.  We describe the code as its set of codewords; a given code might support several different encoding and decoding algorithms.

The types of errors that are expected to occur are described by a \emph{channel model}.  In DRAM, single-bit errors are common, as are errors that affect several bits sharing a single element in the circuit such as a partial or complete row or column \cite{beigi2023systematic}.  These errors may affect multiple symbols, much like a burst error model, but where the bursts tend to occur in \emph{aligned} groups: for example, affecting multiple adjacent bits within a DQ or multiple adjacent DQs within a chip.  We therefore suggest organizing the codeword so that each DQ transmits one or more symbols, as with Bamboo codes \cite{kim2015bamboo}.  For simplicity, some of our calculations instead assume a \emph{symmetric} channel, in which any symbol value is equally likely to be corrupted into any other symbol value.

Usually the decoder does not know which symbols have errors.  But some serious DRAM failure modes are persistent, so the decoder might know (for example) that a certain chip's data is likely corrupted.  This condition is called an \emph{erasure} rather than an error, and is easier to correct than an error.  RS codes can correct combinations of errors and erasures.


\subsection{Beyond-bound decoding}

Each $(n,k)$ code \code has some minimum distance of $d \leq n-k+1$ symbols between its codewords.  With symmetric channels, each decoding algorithm has some threshold
$$t\leq \lfloor (d-1)/2 \rfloor\leq \lfloor (n-k)/2 \rfloor$$
which is the maximum number of symbol errors that it can reliably decode.  In other words, to correct $t$ symbol errors with perfect reliability,
at least $2t$ symbols must be added \cite{komamiya1953singleton}.  Here we will take $t := \lfloor (d-1)/2\rfloor$ unless otherwise indicated.

Nonetheless, a \emph{beyond-bound decoder} can be designed to correct some number $u>t$
of symbol errors.  Such a decoder must fail to uniquely correct some patterns of $u$ symbol errors, because
some blocks will be within distance $u$ of at least two codewords in $\mathcal{C}$, but the failure probability might be quite low.  

If more than $t$ errors have occurred, then the decoder might miscorrect if there is a valid codeword within distance $t$, causing SDC. Decoding more errors typically implies a higher miscorrection probability, simply because more potential ways to correct also means more potential ways to miscorrect.  In some applications, we might even want to throw out ``successful'' corrections if the channel model considers them unlikely error patterns, to reduce the chance of SDC. 

\subsection{Generalized Reed-Solomon codes}

Reed-Solomon codes \cite{reed1960polynomial} are the best known family of symbol-oriented block codes.
Different authors present these codes in different ways.  We adopt the following \emph{dual} view of Reed-Solomon codes, which is equivalent to the usual presentation up to a change in the multipliers.

Let $k<n\leq q$ be positive integers, and let $\F := \F_q$ be the finite field with $q$ elements.  A \emph{Generalized Reed-Solomon code} $\code\in\GRS(q; n,k)$ is additionally parameterized by a sequence of $n$
distinct labels $\seq{\alpha_i}\in \F^n$,
and by nonzero multipliers $\seq{m_i}\in\F^n$, both numbered from $i=1$ to $n$.  The code $\code$ is a subspace of $\F^n$: its codewords are the kernel of
the \emph{syndrome matrix} $S$, which is the product of an $(n-k)\times n$ \emph{Vandermonde matrix} and an $n\times n$ diagonal matrix:
\begin{align*}
V_{(n-k)\times n, \seq{\alpha_i}} &:= \left(\begin{array}{ccc}
1 & \ldots & 1 \\
\alpha_1 & \ldots & \alpha_n \\
\vdots & \ddots & \vdots \\
\alpha_1^{n-k-1} & \ldots & \alpha_n^{n-k-1}
\end{array}\right),
\\\\D_{\seq{m_i}} &:= \left(\begin{array}{ccc}
m_1 &  &  \\
&  \ddots & \\
 &  & m_n
\end{array}\right),
\\\\\text{Syndrome matrix}\ S &:= V_{(n-k)\times n, \seq{\alpha_i}} \cdot D_{\seq{m_i}}
\\\text{RS code}\ \code &:= \ker S.
\end{align*}
The conventional choice is codes which are \emph{narrow-sense}, meaning that  $m_i = \alpha_i$, and
\emph{primitive}, meaning that $\alpha_i = \alpha^i$ for some generator $\alpha$ of $\F^\ast$.
Our URS codes are not primitive. 
The multipliers can be $m_i$ can be changed by scaling each symbol after encoding and before decoding; without loss of generality we set $m_i:=1$ so that we can omit them.

For fixed multipliers $\seq{m_i}$, the labels $\seq{\alpha_i}$ can be adjusted by any invertible affine
transform over $\F$: that is, replacing $\seq{\alpha_i}$ with $\seq{b\alpha_i+c}$ for any $b,c\in\F$ with
$b\neq 0$ gives the same code\footnote{This turns out to be a special case of unraveling with $\ell=1$.}.  So the labels are not uniquely
determined by the set of codewords.  This  observation can make correction more
efficient by allowing us to optimize the syndrome calculation -- for example, by choosing $\alpha_1=0$
and $\alpha_2=1$.

\subsubsection{The key equation}

\GRS codes have distance $d=n-k+1$, and support efficient algorithms to decode up to the threshold $t = \lfloor (n-k)/2 \rfloor$ errors.
The most well-known family of decoding algorithms descends from Peterson-Gorenstein-Zierler \cite{gorenstein1960two}.  This approach
takes as input a block $B\in\F^n$ and begins by computing a \emph{syndrome} $\sigma := S\cdot B\in \F^{n-k}$.
By definition the syndrome is zero if and only if $B\in\code$.  The syndrome is treated as a polynomial $$\sigma(x) := \sum_{i=0}^{n-k-1} \sigma_i x^i.$$
The hardest step of decoding is locating the errors.  To
determine the set $E$ of error locations, the decoder searches for an \emph{error locator polynomial} $\Lambda$ (or \emph{locator} for short), whose degree $e$ equals the number of  errors, putatively of the form
$$\Lambda(x) := \prod_{i\in E} (x-\alpha_{i})$$
which encodes those error locations as its roots.  We also use the reversed locator, which has the same coefficients as $\Lambda$ but in reverse order: \footnote{Different texts name either $\Lambda$ or $\bar\Lambda$ as the locator. Our choice avoids the issue that when $0\in E$, the reversed locator $\bar\Lambda$ has degree only $e-1$.} 
$$\bar\Lambda(x) := \prod_{i\in E} (1-x\alpha_{i}).$$ Only a subset of possible locators will be consistent with the observed syndrome $\sigma$.  The locator must satisfy a \emph{key equation}:
$$\Omega(x) :\equiv \bar\Lambda(x) \cdot \sigma(x)\ \ \text{mod}\ x^{n-k}$$
where $\Omega$ must have degree less than $e$.
For any given error count $e$, the key equation is affine over $\F$, with $e$ variables and $n-k-e$
constraints.
It can therefore be solved with a generic linear equation solver such as Gaussian elimination, but a dedicated approach
such as Berlekamp-Massey is usually more efficient \cite{berlekamp1966non,massey1969shift}.

Once a candidate error locator polynomial has been determined, the error positions $\alpha_i$
are found by solving for its roots using a method such as Chien search \cite{chien1964cyclic}, and the error magnitudes
can be found using methods such as Forney's algorithm \cite{forney1965decoding}.  If too many errors have occurred, the
key equation might have no solution, or might produce an incorrect locator.  Usually an incorrect locator does not split into distinct roots over the symbol labels $\{\alpha_i\}$, which allows the decoder to detect that an uncorrectable error has occurred.

The main obstacle to extending these decoders beyond bound is that, when $e > \lfloor (n-k)/2 \rfloor$,
the key equation becomes underdetermined: with $n-k-e$ constraints on $e$ unknowns, it generically has $|\F|^{2e+k-n}$ solutions.  We can decode by enumerating these candidate error locator polynomials, determining the roots of each one until a locator is found that splits over the symbol labels $\{\alpha_i\}$, but this takes exponential time in $2e+k-n$.  For sufficiently low-rate codes (with $n \geq 3k$), instead Sudan's algorithm \cite{sudan1997decoding} can extend beyond-bound algebraically, but we are interested in a high-rate setting.

\subsection{Interleaved Reed-Solomon codes}

When errors are anticipated to occur in bursts, it is efficient to form a larger code by interleaving codewords of several
smaller ``row'' codes \cite{bleichenbacher2003decoding}.   The interleaved codewords can be corrected either independently, or jointly to better decode burst errors that affect similar positions in each row.  


\begin{figure}[tb]
\begin{center}
\begin{tikzpicture}
  \tikzstyle{Ravel} = [draw=black,   minimum width=1cm, minimum height=1.4cm, node distance=1cm]
  \tikzstyle{Text} = [draw=none, align=center, minimum width=3cm, minimum height=0.7cm]
  \tikzstyle{MiniText} = [draw=none, align=center, minimum width=0.3cm, minimum height=0.6cm, node distance=0.5cm]
  \tikzstyle{Unravel} = [draw=black, minimum width=1cm, minimum height=0.6cm,node distance=1cm]
  \tikzstyle{Dots} = [draw=black, minimum width=1.5cm, minimum height=1cm, node distance=1.25cm]
  \tikzstyle{Underbrace} = [draw=black,  decoration={brace, mirror, raise=0.9cm},  decorate]
  \tikzstyle{Leftbrace} = [draw=black,  decoration={brace, mirror, raise=0.1cm},  decorate]

  \node[Unravel,name=a01]{$\alpha_{1}$};
  \node[Unravel,name=a02,right of=a01]{$\alpha_{2}$};
  \node[Dots,name=a0dots,right of=a02,minimum height=0.6cm]{$\cdots$};
  \node[Unravel,name=a0k,right of=a0dots,node distance=1.25cm]{$\alpha_{k}$};
  
  \node[Unravel,name=a11,below of=a01, node distance=0.8cm]{$\alpha_{1}$};
  \node[Unravel,name=a12,right of=a11]{$\alpha_{2}$};
  \node[Dots,name=a1dots,right of=a12,minimum height=0.6cm]{$\cdots$};
  \node[Unravel,name=a1k,right of=a1dots,node distance=1.25cm]{$\alpha_{k}$};
  \node[Unravel,name=a21,below of=a11, node distance=0.8cm]{$\alpha_{1}$};
  \node[Unravel,name=a22,right of=a21]{$\alpha_{2}$};
  \node[Dots,name=a2dots,right of=a22,minimum height=0.6cm]{$\cdots$};
  \node[Unravel,name=a2k,right of=a2dots,node distance=1.25cm]{$\alpha_{k}$};
  
  \node[MiniText,name=c0,left of=a01,node distance=1cm]{$\code_0$:};
  \node[MiniText,name=c1,left of=a11,node distance=1cm]{$\code_1$:};
  \node[MiniText,name=c2,left of=a21,node distance=1cm]{$\code_2$:};
  
  \draw[rounded corners,red!60!black] ([shift={(-0.33cm,0.2cm)}] a01.north) rectangle ([shift={(0.3cm,-0.2cm)}] a21.south) {};
  
  \node[Text,below of = a21,node distance=0.9cm,minimum height=0cm]{$\Lambda(x) = x - \alpha_1$};
  
  \draw [Leftbrace] (c0.north west) -- (c2.south west);
  
  \node[Text,name=t1,left of=a11,node distance=2.6cm]{
  	Interleaved\\ codes with \\ the same\\ labels.
};
\end{tikzpicture}
\caption{Interleaving with $\ell=3$.  The $\alpha_i$ in a code position is the label of that symbol, not its value.  A column error's locator must satisfy the key equation of all $\ell$ subcodes.}\label{fig:interleave}
\end{center}
\end{figure}

An \emph{interleaved Reed-Solomon} (IRS) code has $\ell$ interleaved row codes $\code_j$, all of which are RS codes with the same length $n$ and same labels $\seq{\alpha_i}$
but perhaps with different data sizes $k_j$.  This forms a larger code \code with parameters $(N,K) := (\ell n, \sum k_j)$.  We write this as
$$\code = \code_1\times \cdots \times \code_\ell$$
or as $\code = \code_1^\ell$ for interleaving $\ell$ copies of the same code $\code_1$.
If we imagine stacking the codewords on top of each other, then each label $\alpha_i$ is
assigned to a column, as shown in \Cref{fig:interleave}.

IRS codes support beyond-bound collaborative decoding when the channel model favors errors affecting
the $\ell$ symbols in a single column at once (e.g.\ failures affecting a channel, a device, etc.)\footnote{Primitive IRS codes can also correct certain unaligned burst errors beyond-bound, but this is more complex, and we will focus on aligned errors.}.  The simple version of
this algorithm corrects $e$
column errors for a total of up to $E := \ell e$ symbol errors, by observing that a single degree-$e$ locator
must simultaneously satisfy the key equations of all the codes.  This gives us $N-K-E$ affine constraints in $e$ unknowns.  Such a system is likely to have a
unique solution when $e \leq \lfloor (N-K)/(\ell+1)\rfloor$, but past the Komamiya-Singleton bound, uniqueness cannot be guaranteed, so decoding is probabilistic.  Solutions can be found by solving these linear equations, or more efficiently using a multiple shift register synthesis algorithm \cite{krachkovsky1998decoding}.

Decoding can be extended slightly further using more complex techniques such as
power decoding \cite{puchinger2017decoding}.  This enables decoding probabilistically out to a radius of
$$R(n,k,\ell) := \floor{n\cdot \left(1 - \left(k/n\right)^{{\ell}/{(\ell+1)}}\right)}$$
columns in polynomial time.  This decoder is too slow to be used in memory systems, but is interesting in other contexts.

An interleaved code has minimum distance equal to the minimum of its row codes.  This means that if an error affects more than $\lfloor(n-k)/2\rfloor$ symbols in only one row (whether due to some effect that tends to cause single-row errors, or just by chance), then the decoder might fail to correct or even miscorrect the error.

A low-rate Reed-Solomon code (i.e.\ one with $n$ many times larger than $k$) can be transformed into an IRS code
 to allow beyond-bound decoding \cite{schmidt2010syndrome}.  Our work also transforms a non-interleaved RS code into an IRS code, but in high-rate setting and using a completely different technique.
 
\subsection{Motivating application: DDR5 RAM, \chipkill and \doublechipkill}

We will take error-correcting DDR5 DRAM as our motivating example.  In DDR5, for each memory read operation, each of the 10 chips being read transmits 2B on each of its 4 DQs, for a total of 8B per chip.  These 80B are used to store 64B of cache line data and 16B of additional data.  So we could use, for example, an $\GRS(2^8; 80, 64)$ code, which can correct up to eight 1B symbol errors and thus achieves \chipkill.  

Now suppose that the memory system also needs to store 1B of metadata alongside the 64B of user data.  This replaces parity check symbols, giving $\GRS(2^{8}; 80, 65)$, which can only correct 7 errors and so doesn't support within-bound \chipkill: a single device failure is beyond its decoding radius.  We therefore want a beyond-bound decoder which can correct single-device errors with high probability, but still has a low probability of miscorrection.  For example, we could try erasure decoding for each device in sequence until it succeeds.  With 10 devices, this might take up to 10 tries, either sequentially or in parallel.

IRS codes could be used to achieve more efficient \bbck, for example by interleaving three $\GRS(2^8; 20, 16)$ codewords and one $\GRS(2^8; 20, 17)$ codeword.  These codes can correct 2 errors and 1 error respectively.  Since decoding cost is roughly quadratic, decoding these four quarter-sized codes requires of about $1/4$ the effort of a full 7- or 8-error correcting code, and decoding can start on the first row before the others have been received.  Beyond-bound decoding these codes is also fast, but it only works if the error pattern affects multiple rows: if by chance two errors affect the last row and nothing else, then they would be uncorrectable even if they are both on one device.  These RS and IRS options are shown in \Cref{fig:bitdiagram}.

\begin{figure}
\begin{tikzpicture}[x=1.6mm,y=1.4mm]
    \draw[rounded corners,gray,very thin,fill=gray!20!white] (-0.65,-0.5) rectangle +(50,18);
\foreach \j in {0,1} {
    \foreach \k in {0,...,3} {
        \foreach \i in {0,...,7} {
            \draw [very thin, draw=black, fill=blue!40!white, step=1,shift={(0.15,0)}] (\i*5+\k,9*\j) grid  +(0.7,8) rectangle +(0,0);
        }
        \foreach \i in {8,...,9} {
            \draw [very thin, draw=black, fill=red!40!white, step=1,shift={(0.15,0)}] (\i*5+\k,9*\j) grid  +(0.7,8) rectangle +(0,0);
        }
    }
    \foreach \k in {0,1} {
        \foreach \i in {3} {
            \draw [very thin, draw=black, fill=blue!60!white, step=1,shift={(0.15,0)}] (\i*5+\k,9*\j) grid  +(0.7,8) rectangle +(0,0);
        }
    }
}

\draw [very thin, draw=black, fill=white!80!black!60!green, step=1,shift={(0.15,0)}] (8*5,0) grid  +(0.7,8) rectangle +(0,0);
\draw [<->] (-1.5,0) -- (-1.5,17) node[midway,above,rotate=90,align=center]{\footnotesize Burst length 16 bits};
\draw [<->] (0,-1.5) -- +(49,0) node[midway,below]{\footnotesize 10 devices};
\draw [<->] (0,-3) -- +(4,0) node[midway,below]{\footnotesize device};
\def\vstride{6.5}
\def\voff{30}
\foreach \k 
[evaluate=\k as \kk using {int(40)}]
    in {0,...,3} {
    \draw[rounded corners,gray,very thin,fill=gray!\kk!white!70!white] (-0.5,\k*\vstride-0.5-\voff) rectangle +(50,5);
    \foreach \j in {0,...,1} {
        \foreach \i in {0,...,7} {
            \draw [very thin, draw=black, fill=blue!\kk!white, step=1, shift={(0,\k*\vstride)}] (\i*5+\j*2+0.2,-\voff) grid  +(1.6,4) rectangle +(0,0);
        }
        \foreach \i in {8,...,9} {
            \draw [very thin, draw=black, fill=red!\kk!white, step=1, shift={(0,\k*\vstride)}] (\i*5+\j*2+0.2,-\voff) grid  +(1.6,4) rectangle +(0,0);
        }
    }
    
    \foreach \j in {0} {
        \foreach \i in {3} {
            \draw [very thin, draw=black, fill=blue!60!white, step=1, shift={(0,\k*\vstride)}] (\i*5+\j*2+0.2,-\voff) grid  +(1.6,4) rectangle +(0,0);
        }
    }
}
\draw [->,blue!40!black] (16,-0.5) -- (16,-6) node[midway,above,rotate=90] {\footnotesize Mix};
\draw [very thin, draw=black, fill=white!80!black!60!green, step=1] (8*5+0*2+0.2,0*5-\voff) grid  +(1.6,4) rectangle +(0,0);
\end{tikzpicture}
\caption{Two possible assignments of bits to RS symbols in DDR5.  User data in \textcolor{white!40!black!60!blue}{blue}, metadata in \textcolor{white!20!black!60!green}{green}, parity in \textcolor{white!40!black!60!red}{red}.
\\\\
Above, a cache line is encoded as a single $\GRS(2^8; 80,65)$ codeword.  Below, it instead uses $\ell=4$ interleaved codewords, the first three from $\GRS(2^8; 20,16)$ and the last from $\GRS(2^8; 20,17)$.  The interleaved code is faster and has straightforward \bbck, but it is more vulnerable to uncorrectable errors or miscorrection.
\\\\
A codeword for the above arrangement can be unraveled to the below one, enabling \bbck without the increased risk of miscorrection.  Since $\ell=4$, the unraveling map operates on groups of four symbols at a time, such as pairs of columns as highlighted in \textcolor{blue!80!black}{saturated blue}.  Perhaps surprisingly, this works even for the group with one metadata symbol and three parities.
}
\label{fig:bitdiagram}
\end{figure}

IRS codes are fast but unreliable: each row code can only correct 1-2 errors, so multi-chip errors may be uncorrectable even if only a few symbols are affected.  Worse, they can miscorrect such errors.  Multi-chip errors are less common than single-chip errors, but not vanishingly rare.  For example, Rowhammer and other read-disturb vulnerabilities can cause sparse errors across multiple chips.  Some implementations, such as AMD's implementation of chipkill \cite{amdchipkill}, mitigate miscorrection using a matching-columns restriction, which requires that any correction affects only a small number of total columns.  This is effective when many rows are corrupted, but not against sparser errors: if only a single row makes a correction, then the matching-columns restriction cannot determine whether it miscorrected.

In the opposite direction, we might be interested in minimizing SDC instead of maximizing speed.  For this, we might instead use a code over a larger field, such as $\GRS(2^{16}; 40, 33)$.  Such a code can only correct three DQ errors, but its miscorrection probability is approximately 8 orders of magnitude lower than $\GRS(2^{8}; 80, 65)$ and 11 orders of magnitude lower than an interleaving of four 8-bit RS codes.

Similar options are available for other memory systems.  For example, DDR4 modules might use an $\GRS(2^8; 72, 64)$ code without metadata.  This system has 4 symbols per device and so can support within-bound \chipkill, but storing metadata again requires \bbck.

More conservative systems require even more powerful error correction, such as double-Chipkill-correct (\doublechipkill) which can tolerate the failure of more than one device.  This can be achieved, for example, using a DDR5 module in lockstep mode, where 128B of user data and 32B of parity are stored across 20 devices instead of 10. Beyond-bound \doublechipkill is more complex: for example, searching through all double-device failures would require trying $\binom{20}{2} = 190$ combinations.

Chip-granularity erasure decoding is another useful feature for high-reliability DRAM systems: if the system detects that a device has failed (e.g.\ because it has already produced many errors) then its data can be treated as erased.  This economizes error-correction symbols: each erasure only requires one parity symbol to correct, compared to errors which each require two parity symbols.

\subsection{Our contribution: Unraveling Reed-Solomon codes}

\begin{figure}[t]
\begin{center}
\begin{tikzpicture}
  \tikzstyle{Ravel} = [draw=black,   minimum width=0.75cm, minimum height=1.4cm, node distance=0.75cm]
  \tikzstyle{Text} = [draw=none, align=center, minimum width=3cm, minimum height=2cm, node distance=2cm]
  \tikzstyle{Unravel} = [draw=black, minimum width=1.5cm, minimum height=0.7cm,node distance=1.5cm]
  \tikzstyle{Decode} = [draw=black, fill=blue!20!white, rounded corners, minimum width=2cm, minimum height=0.7cm,node distance=2cm]
  \tikzstyle{Dots} = [draw=black, minimum width=1.2cm, minimum height=1cm, node distance=1.35cm]
  \tikzstyle{Underbrace} = [draw=black,  decoration={brace, mirror, raise=0.9cm},  decorate]
  \tikzstyle{Leftbrace} = [draw=black,  decoration={brace, mirror, raise=0.15cm},  decorate]
  \tikzstyle{Rightbrace} = [draw=black,  decoration={brace, raise=0.2cm},  decorate]
  
  \node[Ravel,name=b11]{$\beta_{11}$};
  \node[Ravel,name=b12,right of=b11]{$\beta_{12}$};
  \node[Ravel,name=b21,right of=b12]{$\beta_{21}$};
  \node[Ravel,name=b22,right of=b21]{$\beta_{22}$};
  \node[Dots,name=bdots,right of=b22, node distance=0.975cm,minimum height=1.4cm]{$\cdots$};
  \node[Ravel,name=bk1,right of=bdots,node distance=0.975cm]{$\beta_{n1}$};
  \node[Ravel,name=bk2,right of=bk1]{$\beta_{n2}$};
  
  \node[Unravel,name=a11,below of=b11,node distance=2.1cm,xshift=0.375cm]{$\alpha_{1}$};
  \node[Unravel,name=a12,right of=a11]{$\alpha_{2}$};
  \node[Dots,name=a1dots,right of=a12,minimum height=0.7cm]{$\cdots$};
  \node[Unravel,name=a1k,right of=a1dots,node distance=1.35cm]{$\alpha_{n}$};
  \node[Unravel,name=a21,below of=a11, node distance=1cm]{$\alpha_{1}$};
  \node[Unravel,name=a22,right of=a21]{$\alpha_{2}$};
  \node[Dots,name=a2dots,right of=a22,minimum height=0.7cm]{$\cdots$};
  \node[Unravel,name=a2k,right of=a2dots,node distance=1.35cm]{$\alpha_{n}$};
  
  \draw [Leftbrace]  (b11.north west) -- (b11.south west);
  \draw [Leftbrace]  (a11.north west) -- (a21.south west);
  
  \node[Text,name=t1,left of=b11]{
  	Unraveling code\\
	$\code\in\GRS(q; 2n, 2k)$
  };
  \node[Text,name=t1,left of=a11,node distance=2.4cm,yshift=-0.45cm]{
  	Interleaved codes\\
	$\code_0 \times \code_0$ where\\$\code_0\in\GRS(q; n, k)$
};
  
  \draw [Underbrace] ([xshift=1mm] b11.west) -- ([xshift=-1mm] b12.east) node[midway,anchor=north,yshift=-0.8cm,name=u1]{};
  \node [name=v1,below=-1mm of a11.north] {};
  \draw [->] (u1) -- (v1) node[midway,anchor=west] {$\mix_1$};
  
  \draw [Underbrace] ([xshift=1mm] b21.west) -- ([xshift=-1mm] b22.east) node[midway,anchor=north,yshift=-0.8cm,name=u2]{};
  \node [name=v2,below=-1mm of a12.north] {};
  \draw [->] (u2) -- (v2) node[midway,anchor=west] {$\mix_2$};

  \draw [Underbrace] ([xshift=1mm] bk1.west) -- ([xshift=-1mm] bk2.east) node[midway,anchor=north,yshift=-0.8cm,name=uk]{};
  \node [name=vk,below=-1mm of a1k.north] {};
  \draw [->] (uk) -- (vk) node[midway,anchor=west] {$\mix_k$};
\end{tikzpicture}
\caption{Unraveling with $\ell=2$.  The $\alpha$ or $\beta$ in the code position is the label of that symbol.  The single $\GRS(2n,2k)$ code can be decoded directly, reliably correcting up to $n-k$ errors in any pattern. Or, it can be converted into two interleaved $\GRS(n,k)$ codes with the same rate.  These can probabilistically decode up $\lfloor\frac43(n-k)\rfloor$ errors if they affect at most $\lfloor\frac23(n-k)\rfloor$ columns, or slightly more with advanced decoding algorithms \cite{puchinger2017decoding}.}\label{fig:unravel}
\end{center}
\end{figure}
 

\begin{figure*}
\centering
\begin{threeparttable}
\begin{tabular}{l|c|cc|cc|c|ccc|c}
 & & \multicolumn{2}{c|}{Correctable\tnote{a}} & \multicolumn{2}{c|}{Approx Pr[miscorrect]\tnote{b}} & Analytic & \multicolumn{3}{c|}{Pr[\bbck uncorrectable]\tnote{d}} & Rel.
\\ Code & Distance & reliable & best & dense & sparse & \bbck\tnote{c} & dense & sparse & weight & area\tnote{e}
\\\hline
    RS & $\best{15}$ & \best{7} & \best{7} & $4.41\e{-8}$ & $4.41\e{-8}$ & no & $3.20\e{-14}$ & $3.20\e{-14}$ & \best{7} & 184\%
\\ URS & $\best{15}$ & \best{7} &\best{7}  & $4.41\e{-8}$ & $4.41\e{-8}$ & \best{yes} & \best{$3.55\e{-15}$} & \best{$3.55\e{-15}$} & \best{7} & 161\%
\\\hline
RS over $\F_{2^{16}}$& $\best{\leq15}$ & 3 & 6  & $\best{3.61\e{-14}}$ & \best{$3.61\e{-14}$} & no & $3.20\e{-14}$ & $3.20\e{-14}$ & \best{$\leq7$} & 174\%
\\ IRS $\ell=2$ & 8 & 3 & 6  & $\best{3.61\e{-14}}$ & $1.44\e{-7}$ & \best{yes} & $3.20\e{-14}$ & $4.66\e{-10}$ & 4 & 106\%
\\ \tnote{f}\ \ URS $\ell=2$ & $\best{15}$ & 3 & 6  & $\best{3.61\e{-14}}$ & \best{$3.61\e{-14}$} & \best{yes} & \best{$3.55\e{-15}$} & \best{$3.55\e{-15}$} & \best{7} & 100\%
\\\hline
 IRS $\ell=4$ & 3 & 1 & 6 & \best{$3.55\e{-14}$} & $3.21\e{-5}$ & \best{yes} & $3.20\e{-14}$ & $7.14\%$ & 2 & 38\%
\\ IRS $\ell=8$ & 2 & 0 & \chipkill & \best{$3.55\e{-14}$} & $3.67\e{-3}$ & \best{yes} & \best{$3.55\e{-15}$} & $25\%$ & 1 & \best{30\%}
\\ \tnote{f}\ \ URS $\ell=8$ & \best{15} & 1 & \chipkill & $\best{3.55\e{-14}}$ & $\best{3.55\e{-14}}$ & \best{yes} & \best{$3.55\e{-15}$} & \best{$3.55\e{-15}$} & \best{7} & 34\%
\end{tabular}
\begin{tablenotes}\footnotesize
\item[a] General byte errors corrected within bound: most that can be corrected either reliably or in the best case (column-aligned).  The codes marked ``\chipkill'' can correct up to 6 bytes within-bound, but they all have to be on the same device, so the codes' \bbck reliability is more relevant.
\item[b] Probability that the code miscorrects a large error with metadata and \bbck enabled.  Dense is if many bytes are corrupted; sparse is if (worst-case) $e$ random bytes are corrupted.  The floor of $3.55\e{-14}$ is the fraction of blocks within a single-device distance of a codeword.
\item[c] Does the algorithm support beyond-bound \chipkill without iterating over all devices?
\item[d] Probability that beyond-bound correction fails, returning uncorrectable, if either the whole device or worst-case $e$ of its bytes are corrupted at random.  Weight is the minimum number of bytes that must be corrupted to cause a failure.
\item[e] Estimated relative area for a single-cycle-throughput core.  See the main text for the many assumptions of this calculation.
\item[f] This is the same URS code, but decoded after unraveling at $\ell=2$ or $\ell=8$, so errors are corrected in 16- or 64-bit chunks instead of bytes.  This corrects fewer errors but is faster and smaller, and has a lower miscorrection probability.  URS with $\ell=4$ is omitted because in the most straightforward configuration it is too similar to URS with $\ell=8$.
\end{tablenotes}
\end{threeparttable}
\caption{Concrete application to error-correcting $\times 4$ DDR5, 512-bit user data and 16-bit metadata, and \bbck.  The best or almost-best values are highlighted in \best{green}.  Uses 8-bit symbols unless stated.  The IRS codes use the matching columns heuristic to reduce miscorrection rate.  The URS code is the best in all categories, except for miscorrection probability against dense errors.  This is only because it can correct more error patterns than the IRS code.}
\label{tab:comparison2}
\end{figure*}

The worst case for IRS codes is when errors are sparse, and the best case --- both for \bbck and for reducing SDC --- is when errors affect an entire column at once.  We observe that errors can be shifted to this pattern by applying a mixing operation to each column (using e.g.\ matrices over $\F$) after IRS encoding, so that the data is stored in memory using a different code $\code$.  When decoding, we would apply the inverse mixing operation, thereby ``unraveling'' \code back into the IRS code $\code_\text{IRS}$.  This converts each symbol error in \code to a column error in $\code_\text{IRS}$, so the total number of symbol errors may increase.  But the number of columns with errors stays the same, so this technique corrects the same pattern of errors when the matching-columns restriction is used.

Mixing can reduce the miscorrection probability, but only if the mixing functions are chosen carefully: for example, if every column is mixed in the same $\F$-linear way, then the code is still weak.  Mixing by itself also does not help to correct multi-device errors, only to detect them.

Our work hinges on a surprising observation: for carefully selected row codes and mixing matrices, as shown in Sections~\ref{subsec:construction}-\ref{subsec:g}, the full-block code \code is itself a generalized Reed-Solomon code.  This guarantees that \code has the maximum possible distance, and gives a decoding algorithm that works well for multi-device errors.  We call such a code \code an \emph{Unraveling Reed-Solomon} (URS) code.  We show this relationship in \Cref{fig:unravel}.


It is often possible to unravel a URS code $\code$ in more than one way.  In particular, over a binary field $\F_{2^b}$ it is possible to unravel the same code into $\ell=2^c$ rows for any $c \leq b$, allowing several decoding algorithms with different characteristics.  For larger $\ell$, independent decoding is faster and collaborative decoding stretches further beyond bound, but success requires the errors to have greater alignment.  This matches the design constraints of error-correcting DRAM, since there are different failure probabilities for bits, DQs, pairs of DQs, and devices.  Unraveling also interacts well with erasures at column granularity, such as devices which are known to have failed: each erased column in the larger code simply becomes an erased symbol in the smaller codes.

We initially designed URS codes to meet the design goals of a hardware decoder for a datacenter memory product.  Our DDR5 decoder is described in \cref{application}.  It supports up to 16 bits of metadata. It can correct up to 4 DQ errors within-bound (3 DQ with metadata) across multiple chips by unraveling to an interleaving of $\GRS(2^8; 40,32)$ and/or $\GRS(2^8; 40,33)$, so it almost is as fast as these codes.  But its robustness is much greater, comparable to a code over a 16-bit field which would be about twice as costly.

Our decoder also supports fast \bbck by unraveling to an IRS code with $\ell=8$.  IRS codes ordinarily cannot correct multi-chip errors and are vulnerable to sparse errors, so they would not have met our reliability goals, but the URS code removes both downsides.  With a conventional RS code, we would instead need to try decoding with each of the 10 chips erased, which would cost a large factor in either area or worst-case throughput, and would also have 9$\times$ the failure probability.  Using a URS code avoids this cost without any significant disadvantage.  Our decoder also supports correcting two DQ errors (one with metedata) after erasing any single device, and this erasure is more efficient than it would be with a conventional RS code.

Unraveling also gives an efficient recursive encoder, similar to a fast Fourier transform (FFT) or additive fast Fourier transform (AFFT \cite{gao2010additive}).  This Fourier-like calculation can also be used in the search step instead of Chien search.  Chien search cannot be used without modification since it applies only to primitive codes, but URS codes are not primitive.

We also show a second decoding algorithm that can correct more complex patterns of errors, though it is somewhat less efficient.  For example, this decoder can perform beyond-bound decoding of errors affecting one device plus any other single symbol (\bbck+1).

We compare URS codes to IRS codes and to other \GRS codes in \Cref{tab:comparison2}.  This table includes an estimate of relative area for the core, compared to a core implementing URS with $\ell=2$.  There are many trade-offs in decoder architecture which have a large impact on area and performance.  For consistency we assumed a low-latency single-cycle-throughput core.  This consists of a single-cycle syndrome; independent-per-row key equation solvers using either explicit formulas (for $e\leq 2$) or two iterations per cycle of the inversionless Euclidean algorithm; root search and error magnitude solved using static matrices as in \cite{katayama2000one}; and parallel division and error correction (when $e > 2$; for $e=2$ the division is common and for $e=1$ none is required).  We assumed that the core can be configured with or without metadata, and that the core supports \bbck with metadata either analytically (when possible) or by computing Forney syndromes with a large matrix and checking whether they are zero.  We did not factor the syndrome or polynomial-evaluation matrices based on (Additive) Fast Fourier Transforms: this can reduce the size of a URS decoder, and may also be possible with other decoders depending on parameters.  The area includes data buffering, but not small components such as error counting and tag tracking.

Since URS codes are a special type of generalized \GRS code, they can correct at least the same patterns of errors, plus they support efficient beyond-bound decoding of column errors.  Since URS codes met our design goals better than previously known codes, we would expect them to be applicable to other DRAM designs, and possibly beyond memory applications.  For higher-reliability DRAM requirements such as \doublechipkill, URS codes have an even greater advantage.

However, URS codes require careful  choices of the symbol labels and mixing matrices.  As far as we know, these only exist for some $(\F,n,\ell)$.  Fortunately, when $\ell=2$, URS codes exist over any finite field $\F_q$ for every $N=2n\leq q$.

Compared to IRS codes, URS codes are able to decode sparse and multi-device errors, and have no special vulnerability to miscorrecting them.  But since they are full-block codes, URS codes must wait until all the data is received before they can detect or correct errors.

\section{Mathematical Contributions} \label{mathcontributions}

\subsection{Construction of Unraveling Reed-Solomon codes}\label{subsec:construction}

Now we will show how to construct URS codes of order $\ell$, with parameters $(N,K) := (\ell n, \ell k+a)$, with data width $k<n$ and a remainder $0\leq a < \ell$.  Choose a polynomial $G(x)$ of degree $\ell$, and a sequence $\seq{\alpha_i}$ of $n$ distinct elements in $\F$ such that for each $\alpha_i$, the polynomial $G-\alpha_i$ splits into $\ell$ distinct roots $\seq{\beta_{ij}}$ over $\F$.  The unraveling code $\code\in\GRS(q;N,K)$ has labels $\seq{\beta_{ij}}$ (with $i$ being the column number) and multipliers $m_{ij}=1$.

For the row codes, let $\code_0 \in\GRS(q; n,k)$ and $\code_1\in\GRS(q; n,k+1)$ have labels $\seq{\alpha_i}$ and multipliers $m_{i}=1$, and
let $\code_0^{\ell-a} \times \code_1^a$
be an interleaving of these codes.  The \emph{unraveling} map
$$\text{unravel}: \code \to \code_0^{\ell-a} \times \code_1^a$$
is a linear map which acts on each column $i$ by the $\ell\times\ell$ Vandermonde matrix $$\mix_i := V_{\ell \times \ell : \seq{\beta_{ij}}}.$$
Each $\mix_i$ is a square Vandermonde matrix with distinct labels, so it is invertible.  Since the codes $\code$ and $\code_0^{\ell-a} \times \code_1^a$ are the same size, the unraveling map is a 1-to-1 map between them.  So to encode to $\code$, we can encode segments of the
data to each $\code_j$ and then apply $\mix^{-1}_i$ to each column $i$.


\begin{thm}\label{thm:unravel}
The map $\mathrm{unravel}$ indeed maps $\code\to \code_0^{\ell-a} \times \code_1^a$.
\end{thm}
\begin{proof}
The proof is similar to  Taylor expansion at $G$ \cite[Section II]{gao2010additive}.
A codeword $C:=\seq{C_{ij}}\in\F^N$ of $\code$ satisfies the syndrome equation:
$$
\forall m\in\{0,\ldots, N-K-1\}:\ \ 
\sum_{j=1}^{\ell}
\sum_{i=1}^{n} C_{ij}\cdot\beta_{ij}^m = 0.
$$
After unraveling, the $h$'th interleaved codeword has symbols $$U_{ih} := \sum_{j=1}^{\ell} C_{ij} \cdot \beta_{ij}^h.$$
Let $k_h := k$ if $h < \ell - a$ and $k_h := k+1$ otherwise; by case analysis, $\ell (k_h+1) \geq K+h+1$.
We want to show that for each $h$, the row $\seq{U_{ih}}$ is an $\GRS(q; n,k_h)$ codeword with labels $\seq{\alpha_i}$, meaning that:
$$
\forall m\in\{0,\ldots, n-k_h-1\}:\ \ 
\sum_{i=1}^{n} U_{ih}\cdot\alpha_i^m
 \stackrel?= 0.
$$
Expanding $U_{ih} = \sum_{j=1}^{\ell} C_{ij} \cdot \beta_{ij}^h$ and $\alpha_i = G(\beta_{ij})$, this is equivalent to:
$$
\forall m\in\{0,\ldots, n-k_h-1\}:\ \ 
\sum_{j=1}^{\ell} 
\sum_{i=1}^{n} C_{ij} \cdot \beta_{ij}^h\cdot G(\beta_{ij})^m
 \stackrel?= 0.
$$
Now $\beta_{ij}^h\cdot G(\beta_{ij})^m$ is a polynomial in $\beta_{ij}$ of the form $\sum_f g_{ijf} \beta_{ij}^f$ with degree $h + \ell m$, and this degree is bounded by:
\begin{align*}
h + \ell m
& \leq h + \ell (n-k_h-1)
\\ & = h + N - \ell (k_h+1)
\\ & \leq h + N - (K+h+1)
\\ & = N - K - 1.
\end{align*}
Therefore
\begin{align*}
\sum_{j=1}^{\ell} 
\sum_{i=1}^{n} C_{ij} \!\cdot\! \beta_{ij}^h \!\cdot\! G(\beta_{ij})^m
&= 
\sum_{f=0}^{N-K-1}
\sum_{j=1}^{\ell} 
\sum_{i=1}^{n} C_{ij} \cdot g_{ijf} \cdot\beta_{ij}^f
\\&= 
\sum_{f=0}^{N-K-1}
 g_{ijf}\!\cdot\!\left(
\sum_{j=1}^{\ell} 
\sum_{i=1}^{n} C_{ij}  \cdot\beta_{ij}^f
\right)
\\ &= 0
\end{align*}
since, for codewords of $\code$, all the inner terms are all zero by definition.  This completes the proof.
\end{proof}

\subsection{Possible values of $G$}\label{subsec:g}

How should we instantiate $G$ of degree $\ell$, and choose labels $\seq{\alpha_i}$ so that $G(x)-\alpha_i$ splits over $\F$?  This seems like a rather strong requirement, and whether it is possible depends on $\F, n$ and $\ell$.  It is sufficient, but not necessary, for $G(x)$ to be an $\ell$-to-1 map over $\F$ or  $\F\backslash\{0\}$, and then the $\seq{\alpha_i}$ can be simply be chosen among the images of $G$.  Call such a map \emph{collapsing}.  We describe two simple families of collapsing maps.

First, if $\ell$ divides $q-1$, then $G(x) := x^\ell$ is obviously collapsing.  Note that if $\ell$ is composite, then the same code can also be unravelled by any factor of $\ell$.

Alternatively, we can construct collapsing maps when $\ell$ is a power of the characteristic of $\F$, using similar techniques to the AFFT (see e.g.\ \cite[Section III-IV]{gao2010additive}, which deals with the $\F_{2^b}$ case only). Suppose $\F$ is an extension of a field $\K$ --- for example, $\F=\F_{2^b}$ and $\K=\F_2$ --- and let $W$ be a vector subspace of $\F$ over $\K$.  Thus $|W|$ is a power of $|\K|$ and of the characteristic of \F.  Let $$G_W(x) := \prod_{w\in W} (x-w).$$
This is a collapsing function, as shown in the following lemma and corollary:
\begin{lem}
Despite being a polynomial of degree greater than 1, $G_W$ is also a $\K$-linear function.
\end{lem}
\begin{cor}
The kernel of $G_W$ is exactly $W$, so it is $|W|$-to-one.  Its degree is also $|W|$, so it is collapsing.
\end{cor}
\begin{proof}
We begin by noting that if $w\in W$ then $G_W(x+w) = G_W(x)$ by rewriting the indices.  We want to show that $G_W(x)$ is a $\K$-linear function, meaning that for all $a\in\K$ and all $x$ and $y$ in $\F$:
$$G_W(ax) \stackrel?= aG_W(x)\ \ \ \ \text{and}\ \ \ \ G_W(x+y) \stackrel?= G_W(x) + G_W(y).$$
The first part is trivial for $a=0$.  For nonzero $a\in\K$ we see that:
\begin{align*}
G(ax) &= \prod_{w\in W} (ax-w) = \prod_{w\in W} a(x-w/a) \\&= \prod_{w\in W} a(x-w) = a^{|W|} G(x)  = a G(x) 
\end{align*}
by Lagrange's theorem.  To prove additivity, consider the polynomial $$H_y(x) := G(x+y) - G(x) - G(y),$$ which we wish to show is zero.  The leading $x$-term cancels out, so $H_y$ has degree (in $x$) strictly less than $|W|$.  But now, for each $w\in W$,
$$H_y(w) = G(w+y) - G(w) - G(y) = G(y) - 0 - G(y) = 0.$$
Since $H_y$ has at least $|W|$ roots but degree less than $|W|$, it must be the zero polynomial.  This completes the proof.
\end{proof}
Thus we can easily construct collapsing $G$ of any degree $\ell\leq|\F|$ that is a power of the characteristic of $\F$.  We also notice that the roots of $G_W-\alpha_i$ are a $W$-coset in $\F$.  It follows that for any subspace $V\subset W$, the map $G_V$ is a $|V|$-to-1 map on those roots.  Therefore, a code constructed using $G_W$ can be unravelled using $G_V$ for any subspace $V$ of $W$.

A particularly nice case arises when $q = 2^{a\cdot 2^b}$, and $\ell = 2^c \leq 2^b$.  In that case, $\F$ has $\K := \F_{2^{2^b}}$ as a subfield.  Let $G_2(x) := G_{\F_2}(x) = x^2 + x$; then it is easily seen that $G_2^b = G_\K$, so it is collapsing, and therefore $G_2^{b'}$ is collapsing for all $b'\leq b$.  These maps $G_2^{b'}$ are sparse maps of degree $2^{2^{b'}}$ with coefficients in $\{0,1\}$, so they are especially efficient to implement.  For example, $G_2^2(x) = x^4 + x$, and is an $\F_4$-linear map.

For every prime power $q$, either $q$ is a power of 2, or else $q-1$ is divisible by 2.  So for every finite field $\F_q$ it is possible to construct a URS code with $\ell=2$.

It is also possible to use a non-collapsing $G$, though the resulting code will have a shorter maximum length.  For example, consider $\F_{2^8}$ and $\ell=6$: the map $G(x) := (x^2-x)^3 - \alpha$ splits for 21 distinct $\alpha$, so it can be used to construct $\mathrm{URS}(N,K)$ with length $N \leq 21\cdot 6 = 126$ symbols.

\subsubsection{Connection with the (A)FFT}
The syndrome matrix of $\code$ is an $(N-K)\times N$ Vandermonde matrix $V$.  We have shown that it has the same kernel $\code$ as $\ell$ interleaved $(n-k)\times n$ Vandermonde matrices times a blockwise mixing matrix; call this $WB$.  Therefore they are row-equivalent, meaning that $V=MWB$ for some invertible matrix $M$.  This matrix $M$ converts from the syndrome of a block with respect to the interleaved codes, to the syndrome with respect to \code.  Because an $(n,k)$ code is a subcode of an $(n,k')$ code for any $k'\leq k$, the matrix $M$ is lower-triangular.  Its coefficients come from the powers of $G$ according to \Cref{thm:unravel}, so as long as $G$'s coefficients are zero-one, so are $M$'s coefficients.

This decomposition is an application of the (Additive) Fast Fourier Transform decomposition, generalized to any polynomial $G$ with the right root structure.  The code \code can be considered as the set of codewords $C$ whose Fourier transform $V\cdot C$ begins with $n-k$ zeros.  The (A)FFT decomposes this transform as $V=MWB$, where $M$ is lower-triangular, $W$ is $\ell$ interleaved smaller Fourier transforms, and $B$ is a layer of ``butterfly'' transforms.  Thus $BC$ is in the IRS code if and only if $WBC$ begins with $n-k$ zeros.  Since $M$ is invertible and lower triangular, this is the same as $VC=MWBC$ beginning with $n-k$ zeros.

The (A)FFT can also be used efficiently to solve the key equation and to find roots of the locator \cite{7543456,lin2014novel,tang2022new,kadir2023efficient}.  The algorithms for solving the key equation are likely no better than Berlekamp-Massey for the values of $n-k$ seen in memory systems, but the benefits for finding roots of the locator may apply.  These techniques are more efficient when the labels $\alpha_i$ of the code are numbered in accordance with the (A)FFT being used, which will typically be the case for unraveling codes.

\subsection{Decoding Unraveling Reed-Solomon Codes}

There are multiple ways to view $\code$, each with its own decoding algorithm. A design might use only one of these decoders, or they can be run either simultaneously or sequentially, perhaps starting with the fastest one.  If multiple decoders are used, then the syndrome need not be recalculated for each one: it can be translated using the matrix $M$.

The most obvious solution is to use the direct decoder: $\code$ is an \GRS code, so it can be decoded as usual out to distance $\lfloor (N-K)/2 \rfloor$.

\subsection{Independently decoding the unraveled codes}
$\code$ can also be decoded by unraveling the codeword to $\ell$ interleaved codes, and then
decoding them independently.  To avoid miscorrections,
solutions should be rejected if they show errors in more than $\lfloor (N-K)/(2\ell)\rfloor$ different columns.  Reed-Solomon decoding has roughly quadratic cost $\Theta(n\cdot (n-k))$ for small codes.  Since an independent decoding algorithm reduces this to $\ell$ copies of $\ell$-times smaller codes, this decoding algorithm is about $\ell$ times faster than general decoding of $\code$.  Reducing the error threshold may also allow faster few-error algorithms to be used, such as Morioka-Katayama \cite{katayama2000one}.

This decoder corrects exactly the same errors that an RS code over $\F_{q^\ell}$ would correct.  This is a strict subset of the errors that independently decoding an IRS code would provide (since IRS would allow different columns in each row), which is a strict subset of the errors that \code can decode as a full-block RS code (since IRS limits the number of errors per row).

The practical SDC rate of a code depends heavily on the set of errors it can correct and on its minimum distance.  Since this URS decoder corrects the same errors that the $\F_{q^\ell}$ code can correct, and it has at least as high a minimum distance (it is maximum distance separable over $\F_q$ and not just over columns), it should have as low an SDC rate as using larger symbols.  The independent IRS decoder could be restricted to the same number of column errors, but it has a much smaller minimum distance, and thus a worse SDC rate.

\subsection{Beyond-bound decoding the IRS code}
After unraveling, we can also apply beyond-bound collaborative decoding to the IRS code \cite{bleichenbacher2003decoding,krachkovsky1998decoding,DBLP:journals/corr/abs-cs-0610074,puchinger2017decoding}.  Note that mixing the columns may increase the number of symbol errors, but it doesn't increase the number of column errors, which is what limits collaborative decoding.  On the contrary, mixing improves collaborative decoding: realistic channels may produce sparse errors, e.g.\ errors affecting only one of the interleaved codewords, which is a weak point of collaborative decoding.  But after mixing, the errors usually will not be sparse.

\subsubsection{Completely unraveling the code}\label{subsubsec:bbck}

In the limit, depending on the field, we may be able to unravel to IRS codes with distances as low as 2 and~3, which correct only one error\footnote{Distance-4 codes unravel to IRS codes of distance~2 and~3, so we can correct these after unraveling further if the field permits.  However, implementations using shorter codes sometimes correct fewer symbols than the half-distance bound, to reduce the probability of miscorrection.  For these use cases we might stop at distance 4, which provides single-error-correct double-error-detect (symbol SECDED) capability.  The same techniques apply in that case.}.  At least one distance-3 code is required to correct.  In our DRAM example, we get the distance-3 code $\code_0\in\GRS(2^8; 10,8)$ and the distance-2 code $\code_1\in\GRS(2^8; 10,9)$.  Distance-2 codes cannot locate errors, so decoding in this case is beyond-bound: they must be decoded collaboratively with distance-3 codes.  But decoding these codes for single-column errors is fast and easy.

Once the syndrome of a distance-3 block is computed as $\sigma(x) := c_1 x+c_0$, the error magnitude must be $c_0$ and its location must be $\alpha_i = c_1/c_0$.
The syndrome of a distance-2 block is the simple parity $\sigma(x) := c_0$ equal to the sum of the error magnitudes, so if one error occurs and its location is known then decoding is trivial.  So decoding the IRS code $\code_0^{\ell-a}\times \code_1^a$ succeeds if and only if:
\begin{itemize}
\item All row codes of distance 3 or greater must report either the same valid error location $\alpha_i$ or no errors, and
\item If at least one of the distance-2 rows reports an error, then at least one distance-3 row must also report an error, so that the location is known.
\end{itemize}
To decode, note that only one field division is needed: the decoder can choose a nonzero syndrome to calculate the location, and then check that the other syndromes are all proportional to it using cross-multiplication.  The code can be decoded in a streaming fashion if desired, starting with the distance-3 rows.

If a failure corrupts a single column, decoding will fail only if all the distance-3 syndromes are zero.  If the corruption is uniformly random among nonzero values, this happens with probability $$(q^{-(N-K-\ell)} - q^{-\ell}) / (1-q^{-\ell}) < q^{-(N-K-\ell)},$$ since after mixing, the errors must be zero in $N-K-\ell$ locations but not all $\ell$ locations (since then there is no error).  Failures of this type are always a detected uncorrectable error, never a miscorrection.  Similarly, the miscorrection probability for uniform multi-column errors is at most $(N/\ell)\cdot q^{-(N-K-\ell)}$.  In the example of DDR-5 with 8-bit metadata, we have $N/\ell=10$ and $N-K-\ell=7$, so the failure and miscorrection probabilities are at most $2^{-56}$ and $10\cdot 2^{-56}$ respectively, which are almost negligible.  Decoding failure requires at least $N-K-\ell+1$ symbols to be corrupted (in this example, all 8 symbols on the chip) so it cannot happen with sparse errors.  This technique is applicable with even more metadata; for example, with 4 bytes of metadata, the same bounds are $2^{-32}$ and $10\cdot 2^{-32}$ respectively, and at least 5 byte errors are required to cause failure.

\subsubsection{Multiple unraveling}
The same code $\code$ can be unraveled with multiple different values of $\ell$.  For our DDR5 example, suppose we have
8 bits of metadata, and so have constructed a code $\code\in\GRS(2^8; 80,65)$ that unravels up to $\ell=8$ ways.  We can directly decode \code to correct up to 7 single-byte failures.  Or we can unravel it at $\ell=2$ to decode at most $\lfloor 15/(2+1)\rfloor = 5$ channel failures affecting up to 10 bytes; at $\ell=4$ to decode at most $\lfloor 15/(4+1)\rfloor = 3$ double-channel bounded faults affecting at most 12 bytes; or at $\ell=8$ to decode at most one single-chip failure affecting up to 8 bytes.  This last option corrects fewer errors than $\ell=2$ or $\ell=4$, but it's much faster and has a much lower failure probability.

In fact, if $\code$ is constructed using $G_W$ as described in \cref{subsec:g}, then it can be unraveled along any subspace of $W$: that is, for the same $\ell$ it can be unravelled in more than one way, to catch errors that occur in different patterns.  For example, in DDR5, each chip has 4 channels which each transmit 16 bits (which we interpret as two eight-bit symbols) per burst.  We might construct an 8-element $W := W_\text{time} \times W_\text{chan}$.  We normally unravel along $W_\text{time}$ to search for errors that affect both the first and second byte from each of several channels.  But we might also be concerned that interference or a power supply glitch could corrupt all the symbols coming from one or more devices at a given time.  Unraveling the same code $\code$ along $W_\text{chan}$ instead searches for errors with this sort of burst structure.

\subsubsection{Alternative view: stereotyped locators}\label{sec:plusone}

An alternative view of the decoding algorithm is to work directly on the code $\code$ without converting it into an IRS code, but to consider locators which are polynomials in $G(x)$ instead of just in $x$.  Polynomials of degree $e$ in $G(x)$ are a linear subspace of the set $P$ of polynomials of degree $e\cdot\ell$ in $x$.  This decoding algorithm is equivalent to collaboratively decoding the row codes by finding a locator that satisfies all their key equations, but it might be less efficient since it isn't as clear how to apply multiple shift register synthesis.

However, it has a further generalization, in that other patterns of errors may also confine the locator to a subspace of $P$.  For example, consider a ``chip failure plus one byte'' error, affecting the 8 bytes on one chip plus any 9th byte.  This can be searched for as a locator of the form $(G(x)-\alpha_i)\cdot(x-\beta_{j})$, where $G(x)$ is a fixed polynomial of degree 8 
--- e.g.\ $G(x) = x^8 + x^4 + x^2 + x$.  Such locators have only 2 nonlinear degrees of freedom ($\alpha_i$ and $\beta_{j}$), and they are confined to an affine subspace of dimension 3 of the full 9-dimensional affine space of monic 9th-degree polynomials, parameterized by $\{\alpha, \beta,\alpha\beta\}$.  Therefore the locator may be found with linear algebra if there are $N-K \geq 9+3=12$ syndrome symbols, whereas 9 general errors would normally require 18 syndrome symbols.

Similarly, ``chip failure plus two bytes'' errors have 3 nonlinear degrees of freedom, and generate locators in a 5-dimensional affine subspace of monic degree-10 polynomials, so they can be found probabilistically with linear algebra if $N-K \geq 15$.

The stereotyping of locators is also useful for erasure correction.  Suppose we have a $\mathrm{URS}(2^8; 40,32)$ code, and we want to provide 1- or 2-error correction after erasing a column of 4 symbols.  Erasing a column of symbols requires multiplying the syndrome by the erasure locator.  With an unraveling code, this locator has a sparse form such as $x^4 + x + \alpha_i$,  and multiplying the syndrome by such a polynomial is much more efficient than a general polynomial multiplication.

\subsubsection{Shortening the code}

Like any code, URS codes can be \emph{shortened}, reducing the number of redundant bits or symbols by fixing some of them to zero.  This is useful when no unraveled code exists over the right field.  So, for example, we might want a long $\F_{2^8}$ code with 7-symbol columns, but (as far as we know) no such unraveling codes exists. Instead 8-symbol columns could be used with the last symbol in the column fixed to 0.  However, unraveling with $\ell=8$ gives only $(N-K)/8$ parity symbols per row instead of $(N-K)/7$.

Likewise, individual bits of a symbol can be fixed to zero, though this makes encoding slightly more complex, requiring a bit matrix instead of a symbol matrix.  This is useful when the number of available bits per column is not a multiple of the desired symbol size.  With $r$ redundant bits over $\F_{2^b}$, this allows 
$N-K = \lfloor r/2^b \rfloor$.

\section{Application in a DDR5 memory controller} \label{application}

\begin{figure*}[t]\centering
\begin{tabular}{r|ccc|ccc|c}
Metadata bits & RS code & Unraveled code & DQ corr. & \bbck code & \chipkill DUE & Weight & Random SDC \\\hline
0 & $(80,64)$ & $(40,32)^2$ & 4 & N/A & 0 & N/A & $4.95\e{-15}$\\
8 & $(80,65)$ & $(40,32)\!\times\!(40,33)$ & 3 & $(10,8)^7\!\times\!(10,9)^{\phantom{2}}$ & $1.39\e{-17}$ & 8 & $1.41\e{-16}$\\
16 & $(80,66)$ & $(40,33)^2$ & 3 & $(10,8)^6\!\times\!(10,9)^2$& $3.55\e{-15}$ & 7 & $3.61\e{-14}$
\end{tabular}
\caption{High-reliability core parameters (all with $q=2^8$) and failure probabilities for different amounts of metadata.  \bbck code: the code shape after unraveling for \bbck.  With no metadata, this is unused because the 4DQ error correction can already reliably correct single-device errors.  \chipkill DUE: the proportion of single-device errors that are uncorrectable.  Weight: the number of bytes on a single device (out of 8) that must be corrupted before the error can be uncorrectable.  Random SDC: the proportion of multi-device errors (corrupting many symbols across many devices) that result in miscorrection and thus silent data corruption.}\label{tab:chipkill}
\end{figure*}

We implemented this technique in a high-throughput, extra-high-reliability 10-device DDR5 error-correction circuit as part of a datacenter memory controller.  This section reports on a preliminary version of the circuit, and should not be considered a product specification.  Our decoder supports additional data integrity features: the datapath is protected with an additional error detection code to mitigate environmental faults in the buffered data during decoding.  To protect against environmental faults in the decoder's calculation, a second circuit can check that the corrected codeword really is a URS codeword and is within the prescribed distance from the buffered data.

We chose unraveling codes because of their very low SDC rate, similar to 16-bit codes but requiring only 8-bit field arithmetic; and to enable efficient \bbck.

Our circuit processes 512-bit data with up to 16 bits of metadata; the amount of metadata is assumed to be set per memory region, so the controller is informed of the amount used in each read.  The code parameters are shown in \cref{tab:chipkill}, along with failure rates as calculated in \cref{subsubsec:bbck} and SDC rates as calculated according to McEliece-Swanson \cite{mceliece1986decoder}.  When no metadata is used, any 4 DQ errors can be corrected.  When metadata is used, only 3 DQ errors or (probabilistically) the 4 DQs on one device can be corrected.  When the number of metadata bits is not a multiple of 8, the remaining bits in the byte function as a checksum to reduce the SDC rate.  Our overall failure and SDC probabilities are extremely low, since the listed probabilities are multiplicative with the rate at which serious DRAM error occur in the first place.

All cases are handled by a unified decoder.  Our decoder does not attempt to decode the large code with 7 or 8 independent byte errors.  Instead it always unravels the code, so that it can decode 3 or 4 errors at a 16-bit granularity.  Even though the $\RS(2^8; 80,k)$ structure is not used directly, it gives robustness because its distance is 15 to 17 bytes, depending on the amount of metadata.

The decoder also supports marking one device as erased (independently per read), in which case it can still correct one DQ error.  Without metadata, this mode can optionally correct 2 DQ errors, but with a much higher probability of miscorrection ($1.47\e{-7}$).

Our decoder is fully pipelined, with a latency of 1 cycle when there is no error, or 5 cycles if there are errors, with single-cycle throughput (64B/cycle) in either case.  We synthesized our design in
TSMC's 6nm process technology
with a target frequency of 1 GHz in a nominal-to-worst-case corner, to match DDR5-8000's throughput of 64 GB/s.  Based on raw synthesis results and conservative place-and-route estimates, synthesized area is roughly equivalent to 254k NAND2 gates, including the data buffer.  The data integrity features (to protect against hardware faults in the decoder, not ordinary uncorrectable errors) are included in the area, but not in the latency because they overlap with the next stage downstream.

\section{Comparison to Other Error Correction Techniques} \label{comparison}

\subsection{Comparison to MUSE-ECC}
In this work we have been assuming a symmetric channel, but in real DRAMs, some errors favor one direction or the other.  If this bias is large enough, the memory controller can use it to help error correction.  MUSE-ECC's residue codes \cite{9923862} implement error correction based on integer arithmetic instead of $\F_{2^n}$ arithmetic, enabling asymmetric codes.  This includes an $(80,67)$ code (measured in bits) that can correct any $1\!\!\!\to\!\!\!0$ 8-bit symbol error.  This exceeds the Komamiya-Singleton bound, which applies only to symmtric channels.

But asymmetric codes are not as applicable in modern DRAMs.  Energy-saving changes reduce the directional bias of errors, and DDR5's on-die ECC can miscorrect, adding an extra bit-flip in either direction.  MUSE-ECC might be adaptable to this model, but this remains to be seen.  We would also like further reliability features such as erasure decoding.

MUSE-ECC also allows symmetric codes with shorter symbols than Reed-Solomon would permit, but since our work focuses on reliability we prefer longer symbols .

\subsection{Comparison to other large-codeword techniques}

Large-codeword error correction generally features low probabilities of silent data corruption, at the cost of greater decoding latency.  We view our work as complementary, rather than competitive, to other large-codeword Reed-Solomon works such as DUO \cite{gong2018duo} and Bamboo ECC \cite{kim2015bamboo}.  These focus mostly on codeword layout and not the mathematics of the correction process, so a new RS code with an improved correction algorithm directly complements them.

Bamboo ECC proposes encoding an entire 64B cache line in a single codeword.  This greatly reduces the SDC rate and increases error correction capability compared to IRS codes, at the cost of a more complex decoder and needing to wait for the entire cache line before decoding.  Switching to a URS code would reduce the cost of the decoder, but not the need to wait for the whole cache line.  For example, with a DDR4 $\RS(72,64)$ code, Bamboo ECC proposes limiting the number of independent errors corrected by the code to e.g.\ either one or two DQ, or a single $\times 4$ device, to reduce SDC.  URS codes directly address this case, by correcting a single $\times 4$ device at a lower cost than general 4-error RS decoding.  The improvement is even more pronounced in \doublechipkill applications, whether the device failures are sequential or may be simultaneous, because these become 2-symbol correction and/or erasure operations instead of 8-symbol operations.

DUO ECC goes further, by repurposing the on-die ECC data to provide better controller-level error correction with a long codeword, at the cost of slightly longer transfers, possibly enabling probabilistic \chipkill or even \doublechipkill for DIMM configurations that would otherwise not support them.  For example, DUO suggests DDR4 $\times 4$ and DDR4 $\times 8$ modes with 18 and 9 devices, 34 and 68 bits per device, respectively, both with a total of $34\cdot 18 = 612$-bit codewords.  These use a $\GRS(2^8; 76,64)$ code which can correct up to 6 byte errors.  Or they can probabilistically correct a single-device error for DDR5, or even a double-device error for DDR4, by attempting decoding with each device (or pair of devices) erased.

As with Bamboo codes, URS codes can complement DUO, but not as effectively because there are not a whole number of symbols per device, which necessitates shortening the codes\footnote{Unless 17-bit symbols are used, but this is worse unless significant metadata is also desired, because it would have a $31\cdot17=527$-bit payload.}.  For example, we could use $\F_{2^9}$ and divide the 100 extra bits into 1 bit of metadata and 11 nine-bit parity symbols.  This would give an $\GRS(2^9; 72, 61)$ code with 4 or 8 symbols per device for DDR4 or DDR5, shortened from 72 or 36 bits to 68 or 34 bits per device, respectively.  In DDR4 this corrects a single-device plus one error efficiently and reliably; a double-device error with failure probability about $2^{-18} \approx 3.81\e{-6}$; or a single-device plus two-symbol error with failure probability about $2^{-9} \approx 0.2\%$.  With DDR5, it can correct a single-device failure except with probability about $2^{-27} \approx 7.45\e{-9}$.  These calculations are all based on analytically solving the key equation rather than iterating over each device.  We could fall back to the iterative solution if the analytic solution fails, which would be a nearly identical solution to DUO but with slightly higher efficiency due to the sparse erasure locator polynomial.

Bennett's \emph{signed parity codes} \cite{https://doi.org/10.48550/arxiv.2301.07271} provide a lightweight \bbck option.  These codes cannot correct double-DQ errors unless the two DQs are on a single device.  Signed parity codes devote one device to simple parity (i.e.\ the xor-sum of the other devices), and part of another device to additional syndrome bits called a \emph{signature}.  A single-device error can be corrected by checking, for each device, whether applying the syndrome's parity to that device would result in a matching signature.  This decoding algorithm is somewhat similar to the fully-unraveled codes in \cref{subsubsec:bbck} which also have a lightweight decoder, and after unraveling (being an interleaved \GRS code with $m_i=1$) also have a full column of simple parity symbols.  However, signed parity codes use a sparse signature matrix, so the decoder can have less area than a URS decoder.

When using beyond-bound correction, there are always ambiguous cases where the decoder must fail to find a unique correction.  URS codes have fewer ambiguous cases than generic linear codes, including signed parity codes.  This is because for any linear code with distance at least two devices, for any pair of devices $D_1 \neq D_2$, there are at least $q^{2\ell+K-N}-1$ possible errors in $D_1$ that alias to $D_2$.  So with a generic code, for each device there are about $(N/\ell-1)\cdot (q^{2\ell+K-N} - 1)$ ambiguous errors, since there are $N/\ell - 1$ other devices.  But with a URS code, there are only $q^{2\ell+K-N}-1$ possible ambiguous errors in a single device, and these same errors alias every other device.  Because there are fewer ambiguous errors, there is a lower probability that \bbck will fail to find a unique solution.  Additionally, URS codes are maximum distance separable in terms of symbols, which gives resiliency if errors may be localized.

\section{Conclusion and future work}

We have shown a way to \emph{unravel} certain Reed-Solomon codes into interleaved codes, to enable beyond-bound decoding using the interleaved codes.  The unraveling process is computationally efficient, and enables practical new beyond-bound decoding schemes for DRAM and other applications.

For future work, we hope to extend this work to other fields and interleaving orders that are not currently supported.  We would like to expand the work on stereotyped locators to cover other patterns with fewer parity symbols.  We also wonder if other codes can support an analog of unraveling, such as bit-oriented codes or codes for asymmetric channels.

\subsection{Acknowledgments}
Many thanks to Tanj Bennett for helpful discussion, and to anonymous reviewers for their feedback.

\subsection{Intellectual property disclosure}
Some of these techniques may be covered by US and/or international patents.

\bibliographystyle{IEEEtranS}
\bibliography{main}

\begin{thebibliography}{10}
\providecommand{\url}[1]{#1}
\csname url@samestyle\endcsname
\providecommand{\newblock}{\relax}
\providecommand{\bibinfo}[2]{#2}
\providecommand{\BIBentrySTDinterwordspacing}{\spaceskip=0pt\relax}
\providecommand{\BIBentryALTinterwordstretchfactor}{4}
\providecommand{\BIBentryALTinterwordspacing}{\spaceskip=\fontdimen2\font plus
\BIBentryALTinterwordstretchfactor\fontdimen3\font minus
  \fontdimen4\font\relax}
\providecommand{\BIBforeignlanguage}[2]{{%
\expandafter\ifx\csname l@#1\endcsname\relax
\typeout{** WARNING: IEEEtranS.bst: No hyphenation pattern has been}%
\typeout{** loaded for the language `#1'. Using the pattern for}%
\typeout{** the default language instead.}%
\else
\language=\csname l@#1\endcsname
\fi
#2}}
\providecommand{\BIBdecl}{\relax}
\BIBdecl

\bibitem{amdchipkill}
{AMD}, ``{BIOS} and kernel developer’s guide ({BKDG}) for {AMD} family 15h
  models 00h-{0F}h processors,''
  \url{https://www.amd.com/system/files/TechDocs/42301_15h_Mod_00h-0Fh_BKDG.pdf},
  2013.

\bibitem{barroso2019datacenter}
L.~A. Barroso, U.~H{\"o}lzle, and P.~Ranganathan, \emph{The datacenter as a
  computer: Designing warehouse-scale machines}.\hskip 1em plus 0.5em minus
  0.4em\relax Springer Nature, 2019.

\bibitem{beigi2023systematic}
M.~V. Beigi, Y.~Cao, S.~Gurumurthi, C.~Recchia, A.~Walton, and V.~Sridharan,
  ``A systematic study of {DDR4} {DRAM} faults in the field,'' in \emph{2023
  IEEE International Symposium on High-Performance Computer Architecture
  (HPCA)}.\hskip 1em plus 0.5em minus 0.4em\relax IEEE, 2023, pp. 991--1002.

\bibitem{https://doi.org/10.48550/arxiv.2301.07271}
\BIBentryALTinterwordspacing
T.~Bennett, ``{Chip Guard ECC}: An efficient, low latency method,'' 2023.
  [Online]. Available: \url{https://arxiv.org/abs/2301.07271}
\BIBentrySTDinterwordspacing

\bibitem{berlekamp1966non}
E.~R. Berlekamp, ``Non-binary {BCH} decoding,'' North Carolina State
  University. Dept. of Statistics, Tech. Rep., 1966.

\bibitem{bleichenbacher2003decoding}
D.~Bleichenbacher, A.~Kiayias, and M.~Yung, ``Decoding of interleaved {Reed
  Solomon} codes over noisy data,'' in \emph{Automata, Languages and
  Programming: 30th International Colloquium, ICALP 2003 Eindhoven, The
  Netherlands, June 30--July 4, 2003 Proceedings 30}.\hskip 1em plus 0.5em
  minus 0.4em\relax Springer, 2003, pp. 97--108.

\bibitem{cheng2022depth}
Z.~Cheng, S.~Han, P.~P. Lee, X.~Li, J.~Liu, and Z.~Li, ``An in-depth
  correlative study between {DRAM} errors and server failures in production
  data centers,'' in \emph{2022 41st International Symposium on Reliable
  Distributed Systems (SRDS)}.\hskip 1em plus 0.5em minus 0.4em\relax IEEE,
  2022, pp. 262--272.

\bibitem{chien1964cyclic}
R.~Chien, ``Cyclic decoding procedures for {Bose-Chaudhuri-Hocquenghem}
  codes,'' \emph{IEEE Transactions on information theory}, vol.~10, no.~4, pp.
  357--363, 1964.

\bibitem{criss2020improving}
K.~Criss, K.~Bains, R.~Agarwal, T.~Bennett, T.~Grunzke, J.~K. Kim, H.~Chung,
  and M.~Jang, ``Improving memory reliability by bounding {DRAM} faults: {DDR5}
  improved reliability features,'' in \emph{The International Symposium on
  Memory Systems}, 2020, pp. 317--322.

\bibitem{dell1997white}
T.~J. Dell, ``A white paper on the benefits of chipkill-correct {ECC} for {PC}
  server main memory,'' \emph{IBM Microelectronics division}, vol.~11, no.
  1-23, pp. 5--7, 1997.

\bibitem{dixit2021silent}
H.~D. Dixit, S.~Pendharkar, M.~Beadon, C.~Mason, T.~Chakravarthy, B.~Muthiah,
  and S.~Sankar, ``Silent data corruptions at scale,'' 2021.

\bibitem{forney1965decoding}
G.~Forney, ``On decoding {BCH} codes,'' \emph{IEEE Transactions on information
  theory}, vol.~11, no.~4, pp. 549--557, 1965.

\bibitem{gao2010additive}
S.~Gao and T.~Mateer, ``Additive fast {Fourier} transforms over finite
  fields,'' \emph{IEEE Transactions on Information Theory}, vol.~56, no.~12,
  pp. 6265--6272, 2010.

\bibitem{gong2018duo}
S.-L. Gong, J.~Kim, S.~Lym, M.~Sullivan, H.~David, and M.~Erez, ``{DUO}:
  Exposing on-chip redundancy to rank-level {ECC} for high reliability,'' in
  \emph{2018 IEEE International Symposium on High Performance Computer
  Architecture (HPCA)}.\hskip 1em plus 0.5em minus 0.4em\relax IEEE, 2018, pp.
  683--695.

\bibitem{gorenstein1960two}
D.~Gorenstein, W.~W. Peterson, and N.~Zierler, ``Two-error correcting
  {Bose-Chaudhuri} codes are quasi-perfect,'' \emph{Information and Control},
  vol.~3, no.~3, pp. 291--294, 1960.

\bibitem{JESD79-4}
{JEDEC Solid State Technology Association}, ``{JEDEC Standard No. 79-4: DDR4
  SDRAM},'' \url{https://www.jedec.org/standards-documents/docs/jesd79-4b},
  2017.

\bibitem{JESD79-5}
------, ``{JEDEC Standard No. 79-5: DDR5 SDRAM},''
  \url{https://www.jedec.org/standards-documents/docs/jesd79-5b}, 2022.

\bibitem{jian2013low}
X.~Jian, H.~Duwe, J.~Sartori, V.~Sridharan, and R.~Kumar, ``Low-power,
  low-storage-overhead chipkill correct via multi-line error correction,'' in
  \emph{Proceedings of the International Conference on High Performance
  Computing, Networking, Storage and Analysis}, 2013, pp. 1--12.

\bibitem{kadir2023efficient}
W.~K. Kadir, H.-Y. Lin, and E.~Rosnes, ``Efficient interpolation-based decoding
  of {Reed-Solomon} codes,'' in \emph{2023 {IEEE} International Symposium on
  Information Theory ({ISIT})}.\hskip 1em plus 0.5em minus 0.4em\relax {IEEE},
  2023, pp. 997--1002.

\bibitem{katayama2000one}
Y.~Katayama and S.~Morioka, ``One-shot {Reed-Solomon} decoding for
  high-performance dependable systems,'' in \emph{Proceeding International
  Conference on Dependable Systems and Networks. DSN 2000}.\hskip 1em plus
  0.5em minus 0.4em\relax IEEE, 2000, pp. 390--399.

\bibitem{kim2015bamboo}
J.~Kim, M.~Sullivan, and M.~Erez, ``Bamboo {ECC}: Strong, safe, and flexible
  codes for reliable computer memory,'' in \emph{2015 IEEE 21st International
  Symposium on High Performance Computer Architecture (HPCA)}.\hskip 1em plus
  0.5em minus 0.4em\relax IEEE, 2015, pp. 101--112.

\bibitem{kim2014flipping}
Y.~Kim, R.~Daly, J.~Kim, C.~Fallin, J.~H. Lee, D.~Lee, C.~Wilkerson, K.~Lai,
  and O.~Mutlu, ``Flipping bits in memory without accessing them: An
  experimental study of {DRAM} disturbance errors,'' \emph{{ACM} {SIGARCH}
  Computer Architecture News}, vol.~42, no.~3, pp. 361--372, 2014.

\bibitem{komamiya1953singleton}
Y.~Komamiya, ``Application of logical mathematics to information theory,'' p.
  437:3, 1953.

\bibitem{krachkovsky1998decoding}
V.~Y. Krachkovsky and Y.~X. Lee, ``Decoding of parallel {Reed-Solomon} codes
  with applications to product and concatenated codes,'' in \emph{Proceedings.
  1998 IEEE International Symposium on Information Theory (Cat. No.
  98CH36252)}.\hskip 1em plus 0.5em minus 0.4em\relax IEEE, 1998, p.~55.

\bibitem{li2022correctable}
C.~Li, Y.~Zhang, J.~Wang, H.~Chen, X.~Liu, T.~Huang, L.~Peng, S.~Zhou, L.~Wang,
  and S.~Ge, ``From correctable memory errors to uncorrectable memory errors:
  what error bits tell,'' in \emph{SC22: International Conference for High
  Performance Computing, Networking, Storage and Analysis}.\hskip 1em plus
  0.5em minus 0.4em\relax IEEE, 2022, pp. 01--14.

\bibitem{7543456}
S.-J. Lin, T.~Y. Al-Naffouri, and Y.~S. Han, ``{FFT} algorithm for binary
  extension finite fields and its application to {Reed–Solomon} codes,''
  \emph{{IEEE} Transactions on Information Theory}, vol.~62, no.~10, pp.
  5343--5358, 2016.

\bibitem{lin2014novel}
S.-J. Lin, W.-H. Chung, and Y.~S. Han, ``Novel polynomial basis and its
  application to {Reed-Solomon} erasure codes,'' in \emph{2014 {IEEE} 55th
  annual symposium on foundations of computer science}.\hskip 1em plus 0.5em
  minus 0.4em\relax IEEE, 2014, pp. 316--325.

\bibitem{9923862}
\BIBentryALTinterwordspacing
E.~Manzhosov, A.~Hastings, M.~Pancholi, R.~Piersma, M.~Ziad, and
  S.~Sethumadhavan, ``Revisiting residue codes for modern memories,'' in
  \emph{2022 55th IEEE/ACM International Symposium on Microarchitecture
  (MICRO)}.\hskip 1em plus 0.5em minus 0.4em\relax Los Alamitos, CA, USA: IEEE
  Computer Society, oct 2022, pp. 73--90. [Online]. Available:
  \url{https://doi.ieeecomputersociety.org/10.1109/MICRO56248.2022.00020}
\BIBentrySTDinterwordspacing

\bibitem{massey1969shift}
J.~Massey, ``Shift-register synthesis and {BCH} decoding,'' \emph{IEEE
  transactions on Information Theory}, vol.~15, no.~1, pp. 122--127, 1969.

\bibitem{mceliece1986decoder}
R.~McEliece and L.~Swanson, ``On the decoder error probability for reed-solomon
  codes (corresp.),'' \emph{IEEE Transactions on Information Theory}, vol.~32,
  no.~5, pp. 701--703, 1986.

\bibitem{meza2015revisiting}
J.~Meza, Q.~Wu, S.~Kumar, and O.~Mutlu, ``Revisiting memory errors in
  large-scale production data centers: Analysis and modeling of new trends from
  the field,'' in \emph{2015 45th Annual IEEE/IFIP International Conference on
  Dependable Systems and Networks}.\hskip 1em plus 0.5em minus 0.4em\relax
  IEEE, 2015, pp. 415--426.

\bibitem{puchinger2017decoding}
S.~Puchinger and J.~Rosenkilde~n{\'e} Nielsen, ``Decoding of interleaved
  {Reed-Solomon} codes using improved power decoding,'' in \emph{2017 IEEE
  International Symposium on Information Theory (ISIT)}.\hskip 1em plus 0.5em
  minus 0.4em\relax IEEE, 2017, pp. 356--360.

\bibitem{reed1960polynomial}
I.~S. Reed and G.~Solomon, ``Polynomial codes over certain finite fields,''
  \emph{Journal of the society for industrial and applied mathematics}, vol.~8,
  no.~2, pp. 300--304, 1960.

\bibitem{DBLP:journals/corr/abs-cs-0610074}
\BIBentryALTinterwordspacing
G.~Schmidt, V.~Sidorenko, and M.~Bossert, ``Collaborative decoding of
  interleaved {Reed-Solomon} codes and concatenated code designs,''
  \emph{CoRR}, vol. abs/cs/0610074, 2006. [Online]. Available:
  \url{http://arxiv.org/abs/cs/0610074}
\BIBentrySTDinterwordspacing

\bibitem{schmidt2010syndrome}
G.~Schmidt, V.~R. Sidorenko, and M.~Bossert, ``Syndrome decoding of
  {Reed-Solomon} codes beyond half the minimum distance based on shift-register
  synthesis,'' \emph{IEEE Transactions on Information Theory}, vol.~56, no.~10,
  pp. 5245--5252, 2010.

\bibitem{schroeder2009dram}
B.~Schroeder, E.~Pinheiro, and W.-D. Weber, ``{DRAM} errors in the wild: a
  large-scale field study,'' \emph{ACM SIGMETRICS Performance Evaluation
  Review}, vol.~37, no.~1, pp. 193--204, 2009.

\bibitem{sudan1997decoding}
M.~Sudan, ``Decoding of {Reed Solomon} codes beyond the error-correction
  bound,'' \emph{Journal of complexity}, vol.~13, no.~1, pp. 180--193, 1997.

\bibitem{tang2022new}
N.~Tang and Y.~S. Han, ``New decoding of {Reed-Solomon} codes based on {FFT}
  and modular approach,'' \emph{arXiv preprint arXiv:2207.11079}, 2022.

\bibitem{udipi2012lot}
A.~N. Udipi, N.~Muralimanohar, R.~Balsubramonian, A.~Davis, and N.~P. Jouppi,
  ``Lot-ecc: Localized and tiered reliability mechanisms for commodity memory
  systems,'' \emph{ACM SIGARCH Computer Architecture News}, vol.~40, no.~3, pp.
  285--296, 2012.

\end{thebibliography}

\end{document}